\documentclass[11pt,a4paper]{article}
\pdfoutput=1

\usepackage[colorlinks=true,urlcolor=Blue,citecolor=Green,linkcolor=BrickRed]{hyperref}
\usepackage[usenames,dvipsnames]{xcolor}
\usepackage[utf8]{inputenc}
\usepackage{algorithm}
\usepackage{algorithmicx}
\usepackage[noend]{algpseudocode}
\usepackage{cite}
\usepackage[OT4]{fontenc}
\usepackage{amsthm}
\usepackage{amssymb}
\usepackage{amsmath}
\usepackage{amsfonts}
\usepackage{todonotes}
\usepackage{authblk}
\usepackage{fullpage}

\title{Slowing Down Top Trees for Better Worst-Case Bounds}
\date{}
\author[1]{Bartłomiej Dudek}
\author[1]{Paweł Gawrychowski}
\affil[1]{University of Wrocław, Poland}

\newcommand{\Oh}{{O}}

\newcommand{\eps}{\varepsilon}
\newcommand{\TT}{\mathcal{T}}
\newcommand{\TD}{\mathcal{TD}}
\newcommand{\TF}{\widetilde{T}}

\newtheorem{theorem}{Theorem}[section]
\newtheorem{lemma}[theorem]{Lemma}

\newcommand{\FIGURE}[4]{
\begin{figure}[#1]
\begin{centering}
\includegraphics[scale={#2}]{{#3}.pdf}
\caption{#4}
\label{fig:#3}
\end{centering}
\end{figure}
}

\begin{document}

\maketitle

\begin{abstract}
We consider the top tree compression scheme introduced by Bille et al. [ICALP  2013] and construct an infinite family of trees on $n$ nodes
labeled from an alphabet of size $\sigma$, for which the size of the top DAG is $\Theta(\frac{n}{\log_\sigma n}\log\log_\sigma n)$.
Our construction matches a previously known upper bound and exhibits a weakness of this scheme, as the information-theoretic lower bound
is $\Omega(\frac{n}{\log_\sigma n})$. This settles an open problem stated by Lohrey et~al. [arXiv 2017], who designed a more
involved version achieving the lower bound. We show that this can be also guaranteed by a very minor modification of the
original scheme: informally, one only needs to ensure that different parts of the tree are not compressed too quickly.
Arguably, our version is more uniform, and in particular, the compression procedure is oblivious to the value of $\sigma$.
\end{abstract}
 
\section{Introduction}

Tree compression with top trees introduced by Bille et al.~\cite{TopTrees} is able to take advantage of internal repeats in a tree
while supporting various navigational queries directly on the compressed representation in logarithmic time. At a high level,
the idea is to hierarchically partition the tree into \emph{clusters} containing at most two boundary nodes that are shared between
different clusters. A representation of this hierarchical partition is called the top tree. Then, the top DAG is obtained by identifying
isomorphic subtrees of the top tree. 
Bille et al.~\cite{TopTrees} proved that the size of the top DAG is always $\Oh(n/\log_\sigma^{0.19}n)$ for a tree on $n$ nodes labeled from
an alphabet of size $\sigma$. Furthermore, they showed that top DAG compression is always at most logarithmically worse
than the classical DAG compression (and Bille et al.~\cite{BilleFG17} constructed a family of trees for which this logarithmic upper
bound is tight). 
Later, H{\"u}bschle-Schneider and Raman~\cite{Hubschle-Schneider15} improved the bound on the size of the top DAG
to $\Oh(\frac{n}{\log_\sigma n}\log\log_\sigma n)$ using a more involved reasoning based on the heavy path decomposition.
This should be compared with the information-theoretic lower bound of $\Omega(\frac{n}{\log_\sigma n})$.

A natural question is to close the gap between the information-theoretic lower bound of $\Omega(\frac{n}{\log_\sigma n})$ and
the upper bound of $\Oh(\frac{n}{\log_\sigma n}\log\log_\sigma n)$. We show that the latter is tight for the top tree construction
algorithm of Bille et al.~\cite{TopTrees}.

\begin{theorem}\label{thm:lower_bound}
There exists an infinite family of trees on $n$ nodes labeled from an alphabet of size $\sigma$ for which size of the
top DAG is $\Omega(\frac{n}{\log_\sigma n}\log\log_\sigma n)$.
\end{theorem}

This answers an open question explicitly mentioned by Lohrey et al.~\cite{LohreyRS17}, who developed a different algorithm for
constructing a top tree which guarantees that the size of the top DAG matches the information-theoretic lower bound.
A crucial ingredient of their algorithm is a partition of the tree $T$ into $\Oh(n/k)$ clusters of size at most $k$, where $k=\Theta(\log_{\sigma}n)$.
As a byproduct, they obtain a top tree of depth $\Oh(\log n)$ for each cluster. Then, a top tree of the tree $T'$ obtained
by collapsing every cluster of $T$ is constructed by applying the algorithm of Bille et al.~\cite{TopTrees}. Finally, all top trees
are patched together to obtain a top tree of $T$. While this method guarantees that the number of distinct clusters is
$\Omega(\frac{n}{\log_\sigma n})$, its disadvantage is that the resulting procedure is non-uniform, and in particular
needs to be aware of the value of $\sigma$ and $n$.

We show that a slight modification of the algorithm of Bille et al.~\cite{TopTrees} is, in fact, enough to guarantee that the number
of distinct clusters, and so also the size of the top DAG, matches the information-theoretic lower bound. The key insight
actually comes from the proof of Theorem~\ref{thm:lower_bound}, where we construct a tree with the property that some of
its parts are compressed much faster than the others, resulting in a larger number of different clusters. The original algorithm
proceeds in iterations, and in every iteration tries to merge adjacent clusters as long as they meet some additional conditions.
Surprisingly, it turns out that the information-theoretic lower bound can be achieved by slowing down this process to avoid
some parts of the tree being compressed much faster than the others. Informally, we show that it is enough to require
that in the $t^\text{th}$ iteration adjacent clusters are merged only if their size is at most $\alpha^{t}$, for some constant $\alpha>1$.
The modified algorithm preserves nice properties of the original method such as the $\Oh(\log n)$ depth of the obtained top tree.

A detailed description of the original algorithm of Bille et al.~\cite{TopTrees} can be found in Section~\ref{se:preliminaries}.
In Section~\ref{se:lower_bound} we prove Theorem~\ref{thm:lower_bound} and in Section~\ref{se:opt} describe the modification.

\section{Preliminaries}\label{se:preliminaries}

In this section, we briefly restate the top tree construction algorithm of Bille et al.~\cite{TopTrees}. The naming convention
is mostly preserved.

Let $T$ be a (rooted) tree on $n$ nodes. The children of every node are ordered from left to right, and every node has a label from
an alphabet $\Sigma$. $T(v)$ denotes the subtree of $v$, including $v$ itself, and $F(v)$ is the forest of subtrees of all
children $v_{1},v_{2},\ldots,v_{k}$ of $v$, that is, $F(v)=T(v_1)\cup T(v_2) \cup \ldots \cup T(v_k)$.
For $1\le s\le r \le k$ we define $T(v,v_s,v_r)$ to be the tree consisting of $v$ and a contiguous range of its children starting
from the $s^\text{th}$ and ending at the $r^\text{th}$, that is, $T(v,v_s,v_r)=\{v\}\cup T(v_s)\cup T(v_{s+1}) \cup \ldots \cup T(v_r)$.

We define two types of \emph{clusters}.
A cluster with only a top boundary node $v$ is of the form $T(v,v_s,v_r)$.
A cluster with a top boundary node $v$ and a bottom boundary node $u$ is of the form $T(v,v_s,v_r)\setminus F(u)$ for a
node $u\in T(v,v_s,v_r)\setminus\{v\}$.

If edge-disjoint clusters $A$ and $B$ have exactly one common boundary node and $C=A\cup B$ is a cluster, then $A$ and $B$ can be \textit{merged}
into $C$. Then one of the top boundary nodes of $A$ and $B$ becomes the top boundary node of $C$ and there are various ways of choosing the bottom boundary node of $C$. See Figure 2 in \cite{TopTrees} for the details of all five possible ways of merging two clusters.

A \textit{top tree} $\TT$ of $T$ is an ordered and labeled binary tree describing a hierarchical decomposition of $T$ into clusters.
\begin{itemize}
 \item The nodes of $\TT$ correspond to the clusters of $\TT$.
 \item The root of $\TT$ corresponds to the whole $T$.
 \item The leaves of $\TT$ correspond to the edges of $T$. The label of each leaf is the pair of labels of the endpoints of its corresponding edge $(u,v)$ in $T$. The two labels are ordered so that the label of the parent appears before the label of the child.
 \item Each internal node of $\TT$ corresponds to the merged cluster of its two children. The label of each internal node is the type of merge it represents (out of the five merging options).
 The children are ordered so that the left child is the child cluster visited first in a preorder traversal of $T$.
\end{itemize}

The top tree $\TT$ is constructed bottom-up in iterations, starting with the edges of $T$ as the leaves of $\TT$.
During the whole process, we maintain an auxiliary ordered tree $\TF$, initially set to $T$.
The edges of $\TF$ correspond to the nodes of $\TT$, which in turn correspond to the clusters of $T$.
The internal nodes of $\TF$ correspond to the boundary nodes of these clusters and the leaves of $\TF$ correspond to a subset of the leaves of $T$.

On a high level, the iterations are designed in such a way that every time a constant fraction of edges of $\TF$ are merged.
This is proved in Lemma 1 of~\cite{TopTrees}, and we describe a slightly more general property in Lemma~\ref{le:shrinking}.
This guarantees that the height of the resulting top tree is $\Oh(\log n)$.
Each iteration consists of two steps:

\paragraph{Horizontal merges.} 
For each node $v\in\TF$ with $k\ge 2$ children $v_1,\ldots,v_k$, for $i=1$ to $\lfloor \frac k2 \rfloor$, merge the edges $(v,v_{2i-1})$ and $(v,v_{2i})$ if $v_{2i-1}$ or $v_{2i}$ is a leaf.
If $k$ is odd and $v_k$ is a leaf and both $v_{k-2}$ and $v_{k-1}$ are non-leaves then also merge $(v,v_{k-1})$ and $(v,v_k)$.

\paragraph{Vertical merges.}
For each maximal path $v_1,\ldots,v_p$ of nodes in $\TF$ such that $v_{i+1}$ is the parent of $v_i$ and $v_2,\ldots,v_{p-1}$ have a single child: If $p$ is even merge the following pairs of edges $\{(v_1,v_2),(v_2,v_3)\},\ldots,\{(v_{p-2},v_{p-1})\}$.
If $p$ is odd merge the following pairs of edges $\{(v_1,v_2),(v_2,v_3)\},\ldots,\{(v_{p-3},v_{p-2})\}$, and if $(v_{p-1},v_p)$ was not merged in the previous step then also merge $\{(v_{p-2},v_{p-1}),(v_{p-1},v_p)\}$.

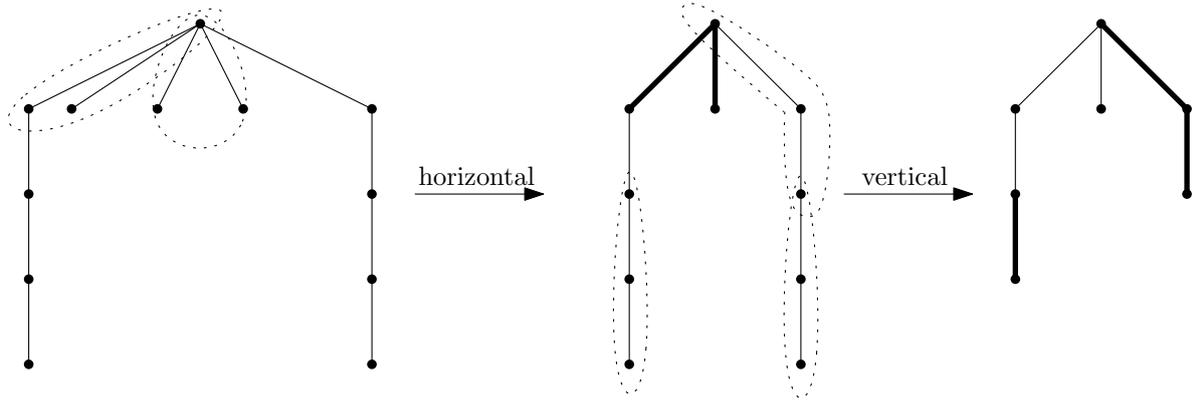
\FIGURE{t}{1.}{example}{Result of a single iteration.
Dotted lines denote the merged edges (clusters) and thick edges denote results of merging.
Note that one edge does not participate in the vertical merge due to having been obtained as a result of a horizontal merge.}

See an example of one iteration in Figure~\ref{fig:example}.
Finally, the compressed representation of $T$ is the so-called top DAG $\TD$, which is the minimal DAG representation of $\TT$
obtained by identifying identical subtrees of $\TT$. 
As every iteration shrinks $\TF$ by a constant factor, $\TT$ can be computed in $\Oh(n)$ time, and then $\TD$ can be computed
in $\Oh(|\TT|)$ time~\cite{DowneyST80}. Thus, the entire compression takes $\Oh(n)$ time.

\section{A lower bound for the approach of Bille et al.}\label{se:lower_bound}

In this section, we prove Theorem~\ref{thm:lower_bound} and show that the $\Oh(\frac{n}{\log_\sigma n}\log\log_\sigma n)$ bound from \cite{Hubschle-Schneider15} on the number of distinct clusters created by the algorithm described in Section~\ref{se:preliminaries} is tight.

For every $k\in \mathbb{N}$ we will construct a tree $T_k$ on $n=\sigma^{8^k}$ nodes for which the corresponding top DAG is of size $\Theta(\frac{n}{\log_\sigma n}\log\log_\sigma n)$.
Let $t=8^k=\log_\sigma n$.
In the beginning, we describe a gadget $G_k$ that is the main building block of $T_k$. 
It consists of $\Oh(t)$ nodes: a path of $t$ nodes and $2^k-1=\Oh(t^{\eps'})$ full ternary trees of size $\Oh(t^\eps)$ connected to the root,
where $\eps+\eps'<1$.
See Figure~\ref{fig:gadget}.
The main intuition behind the construction is that full ternary trees are significantly smaller than the path, but they need the same number of iterations to get compressed.	

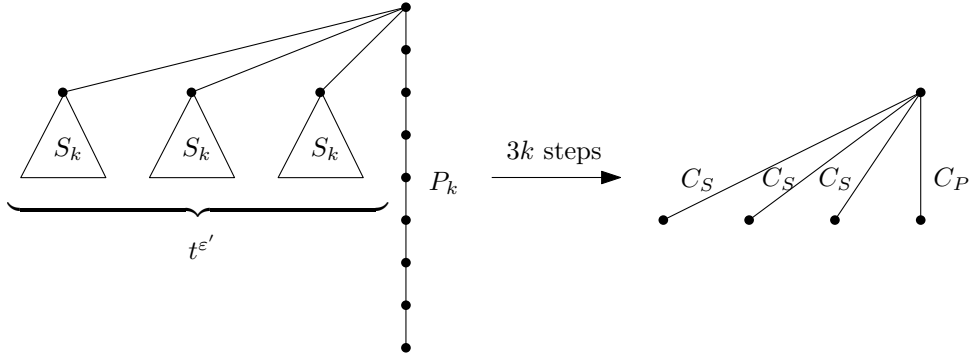
\FIGURE{t}{1.}{gadget}{Gadget $G_k$ consists of $2^k-1=t^{\eps'}$ trees $S_k$ and one path $P_k$. After $3k$ iterations it is compressed to a tree with $2^k$ nodes connected to the root.}

More precisely, let $P_k$ be the path of length $8^k=t$. Clearly, after $3$ iterations it gets compressed to $P_{k-1}$, and so after $3k$ iterations becomes a single cluster.
Similarly, let $S_k$ be the full ternary tree of height $k$ with $3^k$ leaves, so $\frac{3^k-1}{2}=\Oh(3^k)=\Oh(t^{0.53})$ nodes in total.
Observe that after $3$ iterations $S_k$ becomes $S_{k-1}$, and so after $3k$ iterations becomes a single cluster.
To sum up, the gadget $G_k$ consists of path $P_k$ of $t$ nodes and $2^k-1 = \Oh(t^{1/3})$ trees of size $\Oh(t^{0.53})$, so in total $\Oh(t)$ nodes.
After $3k$ iterations $G_{k}$ consists of $2^{k}-1$ clusters $C_{S}$ corresponding to $S_{k}$ and one cluster $C_{P}$ corresponding to $P_{k}$,
as shown in Figure~\ref{fig:gadget}.
In each of the subsequent $k$ iterations, the remaining clusters are merged in pairs. 

Recall that the top DAG contains a node for every distinct subtree of the top tree, and every node of the top tree corresponds to
a cluster obtained during the compression process. In every gadget $G_{k}$ we assign the labels of the nodes of $P_{k}$ so that
the cluster $C_{P}$ obtained after the first $3k$ iterations corresponds to a distinct subtree of the top tree.
Consequently, so does the cluster obtained from $C_{P}$ in each of the subsequent $k$ iterations.

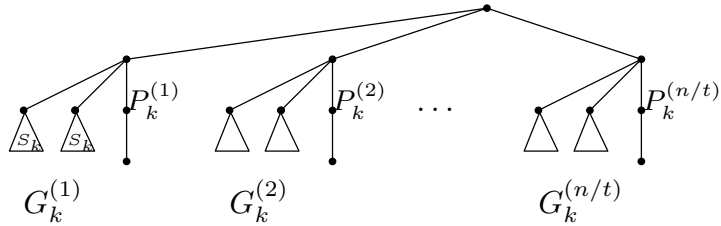
\FIGURE{b}{1.2}{tree_tk}{$T_k$ consists of $\Theta(n/t)$ gadgets $G_k^{(i)}$, where the $i^\text{th}$ of them contains a unique path $P_k^{(i)}$.}

Finally, the tree $T_k$ consists of $\Theta(n/t)$ gadgets connected to a common root as in Figure~\ref{fig:tree_tk}.
The $i^\text{th}$ gadget $G_k^{(i)}$ is a copy of $G_{k}$ with the labels of $P_{k}^{(i)}$ chosen as to spell out the 
the $i^\text{th}$ (in the lexicographical order) word of length $t$ over $\Sigma$.
Note that $\sigma^t> n/t$, so there are more possible words of length $t$ than the number of gadgets that we want to create.
Then each $C_P^{(i)}$ and the clusters obtained from it during the $k$ iterations corresponds to a distinct subtree of the top tree.
Thus, overall the top DAG contains $\Omega(n/t\cdot k)=\Omega(n/t\cdot \log t) = \Omega(n/\log_\sigma n \cdot \log \log_\sigma n)$
nodes, which concludes the proof of Theorem~\ref{thm:lower_bound}.

\section{An optimal tree compression algorithm}\label{se:opt}

Let $\alpha$ be a constant greater than $1$ and consider the following modification of algorithm \cite{TopTrees}.
As mentioned in the introduction, intuitively we would like to proceed exactly as the original algorithm, except that
in the $t^\text{th}$ iteration we do not perform a merge if one of the participating clusters is of size larger than $\alpha^t$.
However, this would require a slight modification of the original charging argument.
To avoid this, in each iteration we first generate all merges that would have been performed in both steps of the original algorithm.
Then we apply only the merges in which both clusters have size at most $\alpha^t$.

\begin{algorithm}[h]
\begin{algorithmic}[1]
  \For{$t=1,\ldots,\Oh(\log n)$}
    \State simulate one iteration of the original algorithm
    \State apply only merges with both clusters of size at most $\alpha^t$
  \EndFor
  \State construct the top DAG $\TD$ of the obtained tree $\TT$
  \State \Return $\TD$
\end{algorithmic}
\caption{A modified top tree construction algorithm of Bille et al.\cite{TopTrees}.}
\label{alg:opt}
\end{algorithm}

Clearly, the depth of the obtained DAG is $\Oh(\log n)$ as before, because after $\log_\alpha n$ iterations the algorithm is no longer constrained and can behave not worse than the original one.
In the following lemma we show that even if there are some clusters that cannot be merged in one step, the tree still shrinks by roughly a constant factor.

\begin{lemma}\label{le:shrinking}
 Suppose that there are $m=p+q$ clusters in the beginning of the $t^\text{th}$ iteration of Algorithm~\ref{alg:opt}, where $q$ is the number of clusters of size larger than $\alpha^t$.
 Then, after the $t^\text{th}$ iteration there are at most $7/8m+q$ clusters.
\end{lemma}
\begin{proof}
The proof is a generalization of the Lemma 1 from \cite{TopTrees}. There are $m+1$ nodes in $\TF$, so at least $m/2+1$ of them
have degree smaller than $2$.
Consider $m/2$ edges from these nodes to their parents and denote this set as $M$.
Then, from a charging argument (see the details in \cite{TopTrees}) we obtain that 
at least half of the edges in $M$ would have been merged in a single iteration of the original algorithm.
Denote these edges by $M'$, where $|M'|\ge m/4$ and observe that at least $|M'|/2 \geq m/8$ pairs of clusters can be merged.

Now, $q$ clusters are too large to participate in a merge. Thus, in the worst case, we can perform at least $m/8-q$
merges.
Thus, after a single iteration the number of clusters decreases to at most $m-(m/8-q) = 7/8m+q$.
\end{proof}

Our goal will be to prove the following theorem.

\begin{theorem}\label{thm:new_alg}
 Let $T$ be a tree on $n$ nodes labeled from an alphabet of size~$\sigma$. Then the size of the corresponding top DAG obtained by Algorithm~\ref{alg:opt} with $\alpha=10/9$ is $\Oh(\frac{n}{\log_\sigma n})$.
\end{theorem}
\noindent
In the following we assume that $\alpha=10/9$, but do not substitute it to avoid clutter.

\begin{lemma}\label{le:bound}
After the $t^\text{th}$ iteration of Algorithm~\ref{alg:opt} there are $\Oh(n/\alpha^{t+1})$ clusters in $\TF$.
\end{lemma}
\begin{proof}
We prove by induction on $t$ that after the $t^\text{th}$ iterations $\TF$ contains at most $cn/\alpha^{t+1}$ clusters,
where $c=113$.
The basis of the induction is immediate.
Consider the $t^\text{th}$ iteration. From the induction hypothesis, after the $(t-1)^\text{th}$ iteration there are at most $c n/\alpha^t$ clusters,
$p$ of them having size at most $\alpha^{t}$ (call them small) and $q$ of them having size larger than $\alpha^{t}$ that
cannot be yet merged in the $t^\text{th}$ iteration (call them big).
We know that $p\leq cn/\alpha^{t}$ and, as the big clusters are disjoint, $q\leq n/\alpha^t$.

We need to show that the total number of clusters after the $t^\text{th}$ iteration is at most $cn/\alpha^{t+1}$.
There are two cases to consider:
\begin{itemize}
 \item $q\le \frac{1}{100} p$: We apply Lemma~\ref{le:shrinking} and conclude that the total number of clusters
 after the $t^\text{th}$ iteration is at most $7/8(p+q)+q < 9/10 p \le cn/\alpha^{t+1}$.
 \item $p<100q$: In the worst case no pair of clusters was merged and the total number of clusters after the $t^\text{th}$ iteration is
 $p+q < 101q < 101 n/\alpha^t\le 113 n/\alpha^{t+1}= cn/\alpha^{t+1}$. \qedhere
\end{itemize}
\end{proof}

\begin{proof}[Proof of \text{Theorem~\ref{thm:new_alg}}]

Clusters are represented with binary trees labeled either with pairs of labels from the original alphabet or one of the $5$ labels representing the
type of merging, so in total there are $\sigma^2+5$ possible labels of nodes in $\TT$.
From the properties of Catalan numbers, it follows that the number of different binary trees of size $x$ is bounded by $4^x$.
Thus there are at most $\sum_{i=1}^x (4(\sigma^2+5))^i \le \sum_{i=1}^x (24\sigma^2)^i \le (24\sigma^2)^{x+1}$ distinct labeled trees of size at most $x$.
Even if some of them appear many times in $\TF$, they will be represented only once in $\TD$.

Consider the situation at the beginning of the $t^\text{th}$ iteration of the algorithm.
Then, from Lemma~\ref{le:bound} there are at most $\Oh(n/\alpha^t)$ clusters in $\TF$.
Setting $t$ such that $\alpha^t+1 = 3/4 \log_{24\sigma^2} n$ we obtain that up to this point at most $n^{3/4}$ distinct clusters of size at most $\alpha^t$ have been created.
As identical subtrees of $\TT$ are identified by the same node in the top DAG, all these clusters are represented by $n^{3/4}$ nodes in $\TD$.
Next, the remaining $\Oh(n/\alpha^t)$ clusters can introduce at most that many new nodes in the DAG.

Finally, size of the DAG obtained by the Algorithm~\ref{alg:opt} on a tree $T$ of size $n$ is bounded by $n^{3/4}+\Oh(n/\alpha^t)=\Oh(n/\log_{24\sigma^2} n)=\Oh(n/\log_\sigma n)$.
\end{proof}

\bibliography{biblio.bib}

\end{document}